\documentclass[letterpaper]{article}
\usepackage{uai2020}
\usepackage[margin=1in]{geometry}

\usepackage{cite}
\usepackage{apacite}
\usepackage{times}
\usepackage{booktabs}
\usepackage{algorithm}
\usepackage{algorithmic}
\usepackage{amssymb}
\usepackage{color}
\usepackage{url}
\usepackage{subfigure}
\usepackage{mathtools}
\usepackage{amsmath}
\usepackage{amsthm}

\usepackage{lipsum} 
\newcommand\blfootnote[1]{%
  \begingroup
  \renewcommand\thefootnote{}\footnote{#1}%
  \addtocounter{footnote}{-1}%
  \endgroup
}

\usepackage{autobreak}
\newcommand{\D}{\mathcal{D}}
\newcommand{\di}{\,\mathrm{d}}
\newcommand{\E}{\mathcal{E}}

\newcommand{\ground}{\mathrm{gt}}

\newcommand{\R}{\mathbb{R}}
\newcommand{\Range}{\mathrm{Range}}
\renewcommand{\S}{\mathcal{S}}
\newcommand{\sampled}{\mathrm{c}}

\newtheorem{theorem}{Theorem}
\newtheorem{definition}{Definition}
\newtheorem{lemma}{Lemma}

\title{Selling Data at an Auction under Privacy Constraints}

\author{ {\bf Mengxiao Zhang} \\
Business School \\
The University of Auckland\\
Auckland, NZ \\
\And
{\bf Fernando Beltran}  \\
Business School \\
The University of Auckland          \\
Auckland, NZ \\
\And
{\bf Jiamou Liu}   \\
School of Computer Science \\
The University of Auckland \\
Auckland, NZ   \\
}

\begin{document}

\maketitle

\begin{abstract}
Private data query combines mechanism design with privacy protection to produce aggregated statistics from privately-owned data records. The problem arises in a data marketplace where data owners have personalised privacy requirements and private data valuations. We focus on the case when the data owners are {\em single-minded}, i.e.,  they are willing to release their data only if the data broker guarantees to meet their announced privacy requirements. For a data broker who wants to purchase data from such data owners, we propose the SingleMindedQuery (SMQ) mechanism, which uses a reverse auction to select data owners and determine compensations.  
SMQ satisfies interim incentive compatibility, individual rationality, and budget feasibility. 
Moreover, it uses {\em purchased privacy expectation maximisation} as a principle to produce accurate outputs for commonly-used queries such as counting, median and linear predictor.  The effectiveness of our method is empirically validated by a series of experiments. 
\end{abstract}

\section{Introduction}
The increasing reliance on data-driven technologies has led to the formation of an economy that is built on data trading\blfootnote{The paper has been accepted at UAI 2020.}. 
Many {\em data marketplaces} emerged that bring {\em data buyers}, i.e., those who are seeking data to purchase, together with {\em data owners}, i.e., those who are willing to release data for a compensation. Examples of data marketplaces include Datacoup, Datum, CitizenMe and DataWallet, many of which enable data buyers to purchase personal data from individual data owners  \cite{laudon1996markets,nget2017balance}. 

Imagine a {\em data broker} who would like to query a set of privately held data records -- such as income record, energy consumption data, or online service rating -- to produce aggregated statistics. 
This may mean that consents  must be purchased from individual data owners to access their data. 
In designing a query mechanism for this task, the data broker faces a number of challenges: 
The first is {\em privacy}. By giving out their data, the data owners give up a certain privacy. An attacker may use, say, income data to infer confidential information such as personal identities \cite{brankovic1999privacy}. It thus makes sense for data owners to demand privacy protection when their data are released. 
Secondly, every data owner associates a value to their data. The value denotes the minimum compensation for the individual to release the data. 
Naturally, this value is hidden from the data broker. The imbalance in the access to data valuation between data owners and the data broker, i.e., {\em information asymmetry}, is the second challenge faced by the data broker. 
Thirdly, when using compensations to incentivise the data owners to release their data, decisions need to be made regarding which (and how many) data owners to procure data from, subject to budget constraints. 
The utility to the data broker is determined by the accuracy of queries made on the purchased data records. The third challenge thus lies in {\em optimising accuracy}. 
These challenges make private data query an attractive topic that sits between data analysis, algorithmic game theory, and data security.

A paradigm of {\em privacy trading} can be found in Ghosh and Roth \citeyear{ghosh2011selling}. The premises involve a data broker who aims to query private data held by individuals who have varying attitudes towards  privacy. To reveal the hidden privacy attitudes of data owners, the data broker uses an auction where each data owner submits a bid reflecting their privacy valuation. Based on all bids received, the data broker decides on a level of privacy to be purchased from the data owners. A noisy query output is then produced which preserves the purchased level of privacy. 
The paradigm laid out in \cite{ghosh2011selling} implicitly makes two assumptions: (1) privacy trading may be done in ways that resemble the trading of other commodities such as stocks and bonds; and (2) the data broker, by giving sufficient incentive, could purchase an arbitrary amount of privacy from every data owner. 
Hence it is up to the data broker to determine the level of privacy to be purchased from the data owners. A number of work have subsequently adopted these views in other settings  \cite{dandekar2012privacy,fleischer2012approximately,ligett2012take,cummings2015accuracy}.

{\bf Contribution.} In this paper, we adopt a different view towards private data query. Under the assumptions of privacy trading, as the data broker applies  more stringent methods to protect data privacy, the cost (and therefore payment) to the data owner would tend towards 0. This is often not the case in reality. In many situations, the data owners are ``single-minded'', i.e., they would demand a level of privacy protection and would not release their data {\em unless their requested level of privacy protection is met}. Moreover, once the level of privacy protection is guaranteed, privacy is decoupled from the cost of releasing the data. In other words, the payment demanded by a data owner to release their data would not decrease as the data broker provides more privacy protection to the data. This paradigm can be considered as {\em data trading} as the purchased commodity is the access to data rather than privacy. To our knowledge,  despite its simple and relatively straightforward setup, no work has addressed private data query under this assumption. This is the major conceptual contribution of our work. 

Our technical contribution lies in  {\em SingleMindedQuery (SMQ)}, a new private data query mechanism over single-minded data owners as described above.  
We adopt the personalised variant of the classical {\em differential privacy} (DP) to quantify data privacy. 
Our goal is to perform queries to a high accuracy while preserving every data owner's declared privacy requirement.  The main technical novelty is a reverse auction mechanism  that determines whose data to purchase and how much compensation should be paid. Note that unlike mechanisms that follow the privacy trading paradigm \cite{ghosh2011selling}, SMQ does not decide the level of privacy to be purchased from the data owners. SMQ incentivises the data owners so that their data valuation is truthfully revealed, thereby resolving information asymmetry.  Furthermore, we use {\em purchased privacy expectation maximisation (PPEM)} as a principle that guides SMQ to achieve high accuracy for commonly-used queries such as counting, median and linear predictor (See Lemma~\ref{accuracy}). We validate empirically the effectiveness of  our method using a series of experiments.


{\bf Related work.} We review research on the query of private data. 
The seminal work of Ghosh and Roth \citeyear{ghosh2011selling} has laid down some main assumptions. 
The authors propose the FairQuery (FQ) mechanism to perform count query on binary (0/1-valued) data. FQ uses a reverse auction to obtain data owners' privacy valuation. When combined with a Laplace mechanism that obfuscates query output, FQ guarantees (certain exact formulations of) incentive compatibility, individual rationality, budget feasibility, query accuracy, and privacy protection.  
%
%
These conditions have since been key indicators of the effectiveness of any query mechanism for private data. The notion of {\em $\varepsilon$-differential privacy (DP)} \cite{dwork2006calibrating} has been chosen to quantify data privacy as the parameter $\varepsilon$ captures in some precise sense the loss on utility a person experiences if her data is used in an $\varepsilon$-DP manner. 
This supports an argument that a data owner's privacy can be regarded as a ``soft constraint'' to be captured by a real-valued cost. The cost increases as a linear function in terms of privacy loss. The FQ mechanism in \cite{ghosh2011selling} associates this cost with the compensation paid to the data owner. 

It is important to point out that, even though the data owners have different privacy valuations, FQ would compute a single level $\varepsilon>0$ of DP, and compensate for $\varepsilon$-DP to {\em all} data owners whose data are used. This means that the mechanism could over-protect some data owners by offering them too stringent privacy protection. Therefore, we can consider this mechanism as ``pseudo-personalised'' as it fails to account for the differences in the data owners' privacy valuations.  Such pseudo-personalised approach has been uptaken by Fleischer and Lyu \citeyear{fleischer2012approximately} and Ligett and Roth \citeyear{ligett2012take} which 
instead of proposing  auction mechanisms, 
design indirect mechanisms, i.e., take-it-or-leave-it offers, to reveal data owners’ privacy valuation. 



In contrast, a ``truly-personalised'' private data query mechanism enables different levels of privacy requirements to be applied to different data owners. Such a mechanism has the potential to avoid over-protecting the data owners' privacy, thus achieving a higher accuracy. 
For example, Dandekar et al. \citeyear{dandekar2012privacy} design the FairInnerProduct (FIP) mechanism for linear predictor queries over real-valued data. FIP also uses a reverse auction mechanism while 
enforcing different levels of privacy protection for different data owners. 
%
Such an approach is also adopted by \cite{cummings2015accuracy}. There, a data broker provides a menu of different variance levels and asks the data owners to report the valuation under each level. 
However, a crucial limitation exists in \cite{cummings2015accuracy}'s mechanism as it lacks theoretical guarantee of differential privacy. 

As mentioned earlier, all mechanisms above follow the privacy trading paradigm where the cost to a data owner from releasing their data is assumed to only arise from privacy loss. This makes sense assuming (1) the data owners fully trust the data broker to protect their purchased privacy level, and (2) the data owners do not have intrinsic valuation to their data records. These assumptions may not hold in practice. In the light of this, we  will put forward a mechanism that compensates the data owners for their intrinsic data valuation while treating the privacy requirement as a hard constraint. 


\section{Preliminaries} \label{sec:prelim}
%
\paragraph*{\bf PDP Queries.}
We regard a {\em dataset} as a tuple $\vec{d}=(d_1,\ldots,d_n)\in \R^n$ where $n\geq 1$ and each {\em data entry} $d_i\in \R$. 
 $\D\subseteq \R^\star$ denotes the collection of all possible datasets.  A {\em query} is a function $\varphi\colon \D\to \R$, 
 such as median and mean. To achieve privacy protection, a randomised function $g$ is applied to the query result to obtain the obfuscated query  $\Phi=g\circ \varphi$. 
%

The notion of {\em $\varepsilon$-personalised differential privacy} (PDP) quantifies the level of privacy achieved by this randomised function: Call two datasets $\vec{d}\in \R^n$ and $\vec{d}'\in \R^n$ {\em $i$-neighbouring} if they differ on exactly the $i$th entry. 
\begin{definition}\cite{jorgensen2015conservative} Given a vector $\vec{\varepsilon}=(\varepsilon_1,\ldots,\varepsilon_n)\in \R^n$, a randomised function $\Phi\colon \D\to \R$ is {\em $\vec{\varepsilon}$-PDP} if for any pair of $i$-neighbouring datasets  $\vec{d}, \vec{d}' \in \R^n$ where $1\leq i\leq n$:
\begin{equation}\label{eqn:PDP}
  \frac{\Pr(\Phi(\vec{d}) \in R)}{\Pr(\Phi(\vec{d}') \in R)} \leq e^{\varepsilon_i}, \forall R \subset \Range(\Phi)
\end{equation}
\end{definition}
In other words, suppose $\vec{d}'$ is an $i$-neighbouring dataset from the true dataset $\vec{d}$. As the ratio above moves closer to $1$, $\Phi$ is more likely to output the same result on $\vec{d}$ and $\vec{d}'$, hiding the true value of the $i$th data entry. Hence a $\varepsilon$-PDP query mechanism with smaller $\varepsilon$ leads to a higher level of privacy  protection for data entry $d_i$\footnote{PDP is generalised from the classical differential privacy \cite{dwork2006calibrating} to accommodate the diversity in people's privacy attitudes \cite{acquisti2005privacy,berendt2005privacy}.}.

The $P\E$ mechanism generates $\vec{\varepsilon}$-PDP queries \cite{jorgensen2015conservative}: 
For $\vec{d}',\vec{d}\in \R^n$,
let $I_{\vec{d} \oplus \vec{d}'} \coloneqq \{1\leq i\leq n\mid d'_i\neq d_i\}$. 
Fix a query $\varphi$. Set $\sigma_{\varphi}(\vec{d},r)\coloneqq \max_{\varphi(\vec{d}')=r} \left\{\sum_{i\in I_{\vec{d} \oplus \vec{d}'}} -\varepsilon_i\right\} \forall r\in \R$. 
Given a dataset $\vec{d}$, the {\em P$\E$ mechanism} $\Phi_{\varphi}(\vec{d})$ generates output $r\in \R$ with  probability
\[
\Pr\left(\Phi_{\varphi}(\vec{d})=r\right)=\frac{\exp(\frac{1}{2} \sigma_{\varphi}(\vec{d}, r))}{\sum_{r'\in \Range(\Phi)} \exp(\frac{1}{2} \sigma_{\varphi}(\vec{d},r'))}
\]
%
We will implement our query using the $P\E$ mechanism, as it can be applied to arbitrary real-valued queries and adds a relatively smaller amount of random noise as compared with other existing methods who claim to achieve $\vec{\varepsilon}$-PDP \cite{alaggan2015heterogeneous,li2017partitioning}.

\smallskip

\noindent {\bf Procurement mechanism. }
We consider a market that consists of a single buyer and multiple sellers denoted by $s_1,s_2,\ldots,s_n$. The following assumptions are made on every seller $s_i$, $1\leq i\leq n$: 

\smallskip

\noindent {\bf (A1)} We assume that once an appropriate amount of compensation is given, $s_i$ is willing to sell her good to the buyer. The required level of compensation depends on the inherent valuation $\theta_i$ that $s_i$ puts on the good. This is a real value in the range $\Theta\coloneqq[\underline{\theta},\overline{\theta}]$ where constants $0\leq \underline{\theta}\leq \overline{\theta}$ are the lower- and upper-bound, respectively; $\theta_i$ represents the loss $s_i$ suffers when she sells the good. 

\noindent {\bf (A2)} The valuation $\theta_i$ of $s_i$ is a random sample from a distribution with cumulative probability function $F_i$ and density function $f_i$. The distribution is assumed to be {\em regular}. In other words, the function $f_i(v)/(1-F_i(v))$, i.e., the probability that $\theta_i=v$ conditioned on $\theta_i>v$, is monotonically non-decreasing on $v\in \Theta$. This assumption is commonly made in mechanism design literature and is satisfied by most distributions \cite{borgers2015introduction}.

We further assume that $F_1=\cdots =F_n$ and $\theta_1,\ldots,\theta_n$ are i.i.d. random variables. The  {\em valuation vector} $\vec{\theta}\coloneqq (\theta_1, \ldots, \theta_n) \in \Theta^n$ has joint distribution and density functions  $F$ and $f$, respectively. While the functions $F$ and $f$ are common knowledge among the buyer and sellers, the valuation $\theta_i$ is only known by $s_i$ and hidden from anyone else. Therefore, it is crucial for the buyer to incentivise the sellers to reveal their true valuations.

%
A {\em procurement mechanism} acts on behalf of the buyer to select  
 a subset of sellers and decides on the amount of compensation for each seller. A direct mechanism, defined below, is a form of procurement mechanism where the buyer makes decisions solely based on the sellers' reported valuations \cite{borgers2015introduction}:

\begin{definition}
A {\em direct mechanism} $\Psi$ consists of a pair of functions $(q,p)$ where $q\colon \Theta^n \to \{0,1\}^n$ is called {\em allocation function} and $p\colon \Theta^n \to \R^n$ is called {\em payment function}. For any $\vec{\psi}\in \Theta^n$, the tuples $q(\vec{\psi})\coloneqq (q_1(\vec{\psi}),\ldots, q_n(\vec{\psi}))$ and $p(\vec{\psi})\coloneqq (p_1(\vec{\psi}), \ldots, p_n(\vec{\psi}))$ are called {\em allocation vector} and {\em payment vector}, resp.
\end{definition}
Intuitively, the buyer  first receives reported valuation $\psi_i\in \Theta$ from each $s_i$ to form a vector $\vec{\psi}=(\psi_1,\ldots,\psi_n)$. The mechanism computes  $q_i(\vec{\psi})$ and $p_i(\vec{\psi})$.  When $q_i(\vec{\psi})=1$, the buyer ``selects'' $s_i$ and purchases the good from $s_i$ with a compensation $p_i(\vec{\psi})$. 
 
We need a procurement mechanism that leads to certain desirable actions of the sellers. At the time of submitting a valuation $\psi_i$, the seller $s_i$ makes decision with only the private information $\theta_i$ and knowledge regarding the distribution $F$.  In other words, the outcome of the mechanism is made based on {\em ex-interim} utility expectations \cite{conitzer2009prediction}. It therefore makes sense to adopt 
{\em Bayesian Nash} utility in the solution concept of mechanism design \cite{mas1995microeconomic}: 

Set $\vec{\psi}_{-i}\coloneqq(\psi_1,\ldots,\psi_{i-1},\psi_{i+1},\ldots,\psi_n) \in \Theta^{n-1}$ as the reported valuation vector of the sellers other than $s_i$. We abuse the notation writing $p_i(\psi_i,\vec{\psi}_{-i})$ for $p_i(\vec{\psi})$ and $q_i(\psi_i,\vec{\psi}_{-i})$ for $q_i(\vec{\psi})$. Set $f_{-i}$ as the density function of the joint probability distribution of $(\theta_1,\ldots,\theta_{i-1},\theta_{i+1},\theta_n)$.
We define 
\begin{equation}
Q_i (\psi_i )\coloneqq \int_{\Theta^{n-1}} q_i(\psi_i,\vec{\psi}_{-i}) f_{-i} (\vec{\psi}_{-i}) \di\vec{\psi}_{-i} 
\label{Qi}
\end{equation}
\begin{equation}
P_i (\psi_i)\coloneqq\int_{\Theta^{n-1}} p_i(\psi_i,\vec{\psi}_{-i}) f_{-i} (\vec{\psi}_{-i} ) \di\vec{\psi}_{-i} 
\label{Pi}
\end{equation}
as the {\em expected allocation} and the {\em expected payment} when the reported valuation of $s_i$ is $\psi_i$, resp. And her {\em expected utility} is
$U_i (\psi_i | \theta_i)\coloneqq P_i (\psi_i )-\theta_i Q_i (\psi_i)$.

The celebrated {\em revelation principle} asserts that to find the optimal procurement process, it is sufficient to restrict to direct mechanisms where the data owners truthfully report their valuation in the Bayesian Nash equilibrium (see \cite{borgers2015introduction}). Formally, we would like to design a direct mechanism $\Psi$ with the following properties:

\smallskip

\noindent {\bf (1) Incentive compatibility (IC):} This property ensures that each seller truthfully reports her valuation, as she expects to gain the maximum utility by doing this, i.e.,
\begin{equation}\label{eqn:IC}
U_i (\theta_i|\theta_i) \geq U_i (\psi_{i}|\theta_i), \forall i \in \{1, \ldots, n\}, \forall \theta_i ,\psi_{i} \in \Theta
\end{equation}
\noindent {\bf (2) Individual rationality (IR):} This property ensures that every seller is willing to participate in the mechanism, as her gain of participating is not less than that of non-participation.
 Here, we assume that the utility of non-participation is zero, i.e.,
\begin{equation}\label{eqn:IR}
U_i (\theta_i |\theta_i) \geq 0 , \forall i \in \{1, \ldots, n\}, \forall \theta_i \in \Theta
\end{equation}

It is also reasonable to assume that the buyer has a limited budget $B \leq \overline{\theta}n$, 
and thus we need the following. 

\noindent {\bf (3) Budget feasibility (BF):}
The expected compensation received by all sellers should not exceed budget $B$, i.e., 
\begin{equation}\label{eqn:BF}
\sum_{i=1}^n\int_{\Theta^n} p_i(\vec{\psi}) f(\vec{\psi})\di\vec{\psi} \leq B
\end{equation}

{\em Remark:} The definition of BF here is {\em interim BF}, which merely considers the expected value of total payment. This may seem counter-intuitive -- as the buyer would really like to ensure ex post BF, i.e., \(\sum_{i=1}^{n}p_i(\vec{\psi}) \leq B\). It has turned out that these two notions are equivalent, as shown by the following lemma.
\begin{lemma}\cite{borgers2015introduction}
For any direct mechanism that is individually rational, incentive compatible and interim budget feasible, there is a direct mechanism with the same allocation rule that is individually rational, incentive compatible and ex post budget feasible.
\label{ex}
\end{lemma}


\section{Privacy-aware data owners}
We consider a data marketplace that involves a group of {\em data owners}, each of whom holds a private (real-valued) data entry, and a  {\em data broker} who would like to collect these data entries. Denote the data owners by $s_1,\ldots,s_n$ and $s_i$'s data entry by $d_i$. Every $s_i$ has a {\em data valuation} $\theta_i\in \R$ for her data entry. In this way, the data owners and broker are respectively the sellers and buyer in the market with data entries as goods.
%
We make some further assumptions regarding every data owner $s_i$:

\smallskip

\noindent{\bf(A3)} $s_i$ is a {\em single-minded} data owner, i.e., the data owner $s_i$ has a {\em privacy requirement} $\varepsilon_i\in \R$. 
To release the data, $s_i$ requires the data broker to meet $\varepsilon_i$-PDP for any query made on the collected dataset.


\noindent {\bf (A4)} The data entry $d_i$, once released, is verifiable and thus data owners cannot misreport their data. This assumption is reasonable because in some data marketplaces such as Datacoup, the data brokers do not directly collect data from data owners, but rather, they seek for data access permission from data owners and the data are provided by certain intermediary service. 

\noindent{\bf (A5)} No correlation exists between $\theta_i$ and the data value $d_i$. This means that the output of the  allocation function does not reveal any information about the private data. 

 (A1)--(A5) naturally infer the following definition:
\begin{definition}
A {\em privacy-aware data owner} $s_i$ is formally a tuple $s_i\coloneqq(d_i,\theta_i,\varepsilon_i)$, where $d_i\in \R$ is $s_i$'s data entry, $\theta_i \in \Theta$ is $s_i$'s {\em data valuation}, and $\varepsilon_i\geq 0$ is her {\em privacy requirement}.
\end{definition}

{\em Remark:} The definition above is  different from the one in \cite{ghosh2011selling} where instead of the value $\varepsilon_i$, a data owner $s_i$ is associated a {\em cost function} $c_i(\varepsilon)\in \R$ that captures the amount of loss $s_i$ experiences when the data broker ``purchases'' $\varepsilon$ amount of privacy. In particular, the cost function is defined as $c_i(\varepsilon)=\varepsilon\cdot v_i$ where $v_i\in \R$ is a {\em privacy valuation}. When using $c_i$ to denote the minimum compensation $s_i$ requires to release the data entry $d_i$, in our setting, a privacy-aware data owner would have a stepwise cost function defined as
\begin{equation}\label{cost}
c_i(\varepsilon)=
\begin{cases}
\theta_i & \text{ if } 0\leq \varepsilon \leq \varepsilon_i\\
\infty & \text{ otherwise}
\end{cases}
\end{equation}
In particular, $c_i(0)$ may be non-zero as it represents $s_i$'s valuation to the data entry rather than privacy. This crucial difference makes the analysis in \cite{ghosh2011selling} not applicable to our setting.

\section{Data query mechanism}


A {\em data query mechanism} $A$ combines a procurement mechanism $\Psi$ with a PDP query $\Phi$ (See Fig.~\ref{Privacy-aware data broker}). Suppose $\S=\{s_1,\ldots,s_n\}$ is a set of privacy-aware data owners. Given a query $\varphi\colon\D\to\R$ and a budget $B>0$, the mechanism $A$ first applies $\Psi$ which purchases data entries from a subset of data owners and constructs a sampled dataset;  the PDP query $\Phi$ is then applied to 
return a final query result $A(\S)$. As argued in Section~\ref{sec:prelim}, we will apply the $P\E$ mechanism as $\Phi$.  
\begin{figure}[h]
  \centering
  \includegraphics[width=0.9\linewidth]{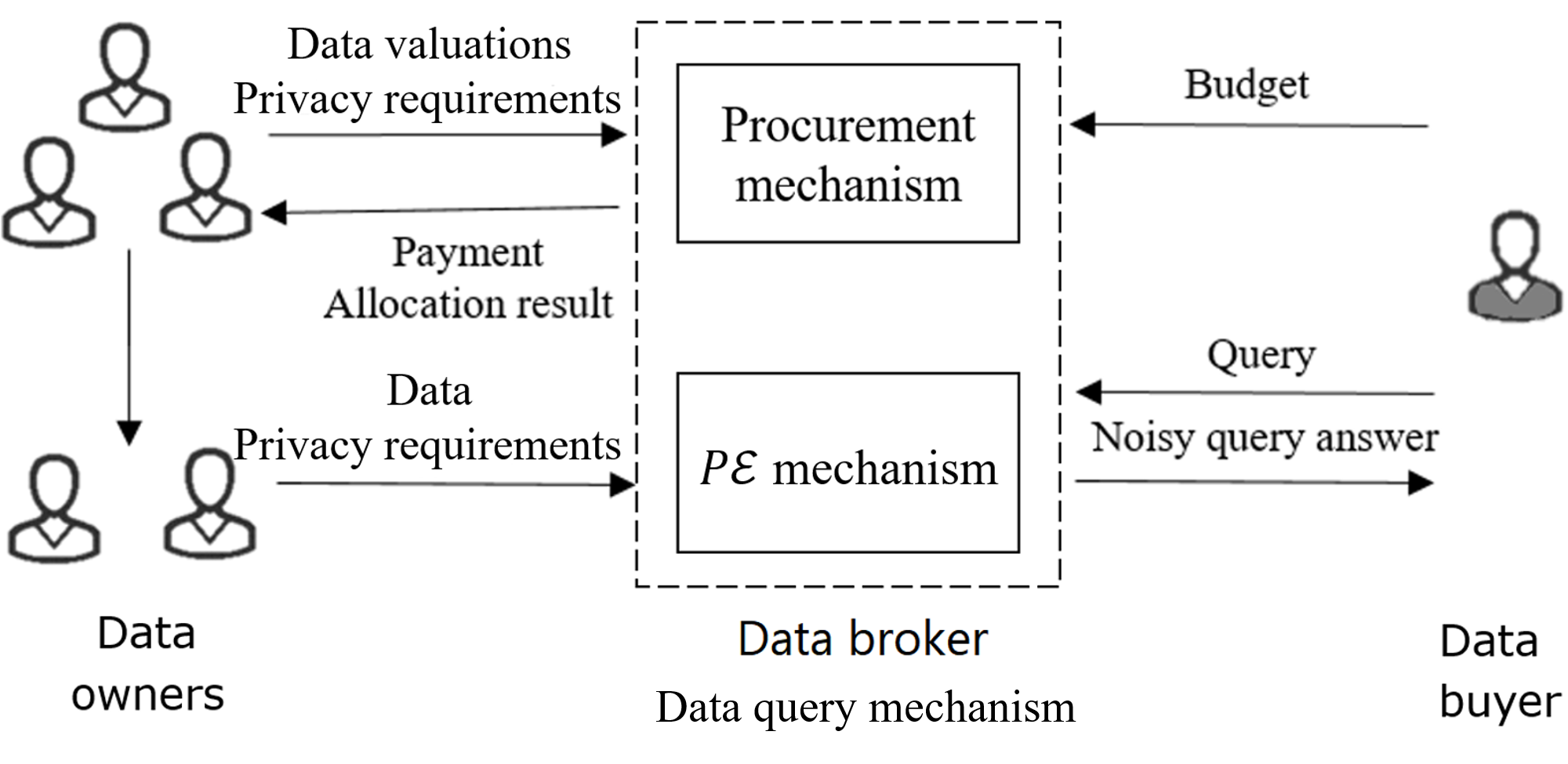}
  \caption{An illustration of a data query mechanism}
  \label{Privacy-aware data broker}
\end{figure}


We denote the {\em ground-truth dataset} $(d_1,\ldots,d_n)$ by $\vec{d}_{\ground}$ and the sampled dataset by $\vec{d}_\sampled$. The next definition captures the query accuracy of $A$. 

\begin{definition}
A query mechanism $A$ is {\em $(\alpha,\delta)$-probably approximately correct} (PAC) if for any $\vec{d}_\ground$, 
$\Pr(|A(\vec{d}_\sampled)-\varphi(\vec{d}_{\ground})|\geq \alpha) \leq 1-\delta$.
\end{definition}

For brevity we write $q_i$ for the allocation result $q_i(\psi_i,\vec{\psi}_{-i})\in \{0,1\}$ of $s_i$ given the reported valuation $\vec{\psi}$. As each data owner is single-minded, $s_i$'s privacy requirement $\varepsilon_i$ must be met in case $q_i=1$. One can thus view $\varepsilon_iq_i$ as the amount of privacy ``purchased'' by the procurement mechanism. We next establishes a connection between the query accuracy 
and the total amount  $\sum_{i=1}^n \varepsilon_iq_i$ of  purchased privacy. 

We consider the following commonly-used query functions: {\em Linear predictor} captures a wide range of potential queries over real numbers that include k-nearest neighbours, Nadaranya-Watson weighted average, ridge regression and support vector machines \cite{dandekar2012privacy}. It is defined as $\varphi(\vec{d}_{\ground})\coloneqq \sum_{i=1}^n w_i d_i$, where $w_i\neq 0$ is the weight of the data owner $s_i$. If $d_i\in \{0,1\}$ and $w_i=1$ for all $i\in \{1,\ldots,n\}$, the query is called a {\em count query}. More generally, a linear predictor can be used in a recommender system where $w_i$ represents the similarity between a user $s_i$ and a new user; $\varphi(\vec{d}_{\ground})$ is then the prediction about the new user's data value. We also consider 
{\em median query} which aims to find the median value among {\em mutually distinct positive integers}. In the next lemma, the query $\varphi$ belongs to one of the three types  defined above.


\begin{lemma}\label{accuracy}
For any integer $1\leq \alpha\leq n/4$ and $\delta\in (0,1)$, if the query mechanism $A$ is $\left(\alpha,\delta\right)$-PAC, then $\sum_{i=1}^n\varepsilon_i q_i\geq \frac{n}{4 \alpha}\cdot (\ln \delta-\ln(1-\delta))$. 
\end{lemma}

\begin{proof} We prove the case when $\varphi$ is the count query. Recall that this case assumes that each data entry $d_i$ is a 0/1-value. 
We assume for a contradiction that $\sum_{i=1}^n\varepsilon_i q_i<\frac{n(\ln \delta-\ln(1-\delta))}{4\alpha}$ and the query mechanism is $(\alpha,\delta)$-PAC. Let $R=\{r \in \R\mid |r-\varphi(\vec{d}_{\ground})|<\alpha \}$. By the definition of $(\alpha,\delta)$-PAC,  
$\Pr\left(\Phi\left(\vec{d}_\ground\right)\in R \right) \geq \delta$.

Sort the data owners so that $\varepsilon_i q_i$ are in ascending order, i.e., $\varepsilon_1 q_1\leq \varepsilon_2 q_2\leq \ldots \leq \varepsilon_n q_n$. Consider the first $ 4\alpha$ data owners (Note that $4\alpha\leq n$). 
Clearly, 
$$\sum_{i=1}^{4\alpha} \varepsilon_i q_i < \frac{n(\ln \delta-\ln(1-\delta))}{4\alpha}\frac{4\alpha}{n}=\ln \delta -\ln (1-\delta).$$
Let $\vec{d}^0 \coloneqq (d_{i})_{i\in I_0}$ and $\vec{d}^1 \coloneqq (d_i)_{i\in I_1}$ where $I_j= \{1\leq i\leq 4\alpha\mid d_i=j\}$ for $j\in \{0,1\}$. 
Without loss of generality, assume that $|\vec{d}^0|>2\alpha$. Let $I'\subseteq I_0$ that contains exactly $2\alpha$ elements, and define a dataset $\vec{d}'\coloneqq (b_1,\ldots,b_n)$ where $b_i=1$ if $i\in I'$, and $b_i=d_i$ otherwise.  
It follows that $\varphi(\vec{d}')=\varphi(\vec{d}_{\ground})+2\alpha$. 

It is straightforward to verify by definition of PDP that 
\begin{align*}
    \Pr\left(\Phi(\vec{d}')\in R\right) \geq & \exp\left(-\sum_{i\in I'}\varepsilon_i q_i\right) \Pr\left(\Phi(\vec{d}_{\ground})\in R\right)  \\
  >   & \exp\left(-(\ln\delta-\ln (1-\delta))\right)\times\delta   \\
  = & \frac{1-\delta}{\delta}\cdot \delta =1-\delta
\end{align*}
Since $\varphi(\vec{d}')=\varphi(\vec{d}_{\ground})+2\alpha $, by the triangle inequality, we have 
$\Pr\left(|\Phi(\vec{d}')-\varphi(\vec{d}')| > \alpha  \right) \geq \Pr\left(|\Phi(\vec{d}')-\varphi(\vec{d}_{\ground})| < \alpha  \right) > 1-\delta$, 
which contradicts the $(\alpha,\delta)$-PAC assumption. 

The proof is similar for the case when $\varphi$ is the general linear predictor where the data entries are real values. The only difference is that we define the set $I'$ as $\{1,\ldots,2\alpha\}$ and   
the dataset $\vec{d}'$ by $b_i=d_i+\frac{1}{w_i}$ for all $i\in I'$ and  $b_i=d_i$ otherwise. The case when $\varphi$ is a median query (over mutually distinct integers) can be found in Appendix A. 
\end{proof}



\section{SingleMindedQuery (SMQ)}\label{sec:SMQ}

Fixing the confidence level $\delta$ ($>1-\delta$), Lemma~\ref{accuracy} asserts that $\sum_{i=1}^n\varepsilon_i q_i \in \Omega(n/\alpha)$ is necessary for any query mechanism to achieve $(\alpha,\delta)$-PAC. This suggests the total amount of purchased privacy $\sum_{i=1}^n\varepsilon_i q_i$ plays a significant role in determining query accuracy. Moreover, since we concern with ex-interim utility, the term $q_i$ here should be considered in expectation, i.e., $\int_\Theta \sum_{i=1}^n \varepsilon_i q_i (\vec{\psi})f(\vec{\psi})\di\vec{\psi}$.
In our query mechanism, we thus aim to {\em maximise purchased privacy expectation} (PPEM) in order to obtain accurate query results, i.e., we aim to solve the following optimisation problem:
\begin{equation}
\begin{aligned}
& \text{maximise}
& & \int_\Theta \sum_{i=1}^n \varepsilon_i q_i (\vec{\psi})f(\vec{\psi})\di\vec{\psi} \\
& \text{such that }
&& \text{\eqref{eqn:IC}, \eqref{eqn:IR} and \eqref{eqn:BF}} \text{ are satisfied } \\
\end{aligned}
\label{OP1}
\end{equation}


In the rest of the paper, we describe our data query mechanism, namely {\em SingleMindedQuery} (SMQ), to solve Problem~\eqref{OP1}. 
%
Note that the problem can be regarded as a {\em knapsack auction problem} (see \cite{ensthaler2014bayesian}): Treat $B$ as the {\em capacity} of the knapsack, $\varepsilon_i$ as the {\em value} and $\theta_i$ as the {\em weight} of the $i$th item ($1\leq i\leq n$).

\begin{definition} A {\em simple direct mechanism} $\Psi$ consists of allocation rule $q_i(\vec{\psi})$ and payment rule $p_i(\vec{\psi})$ as follows,
\begin{equation} 
q_i (\vec{\psi} )=Q_i(\psi_i) \coloneqq
\begin{cases}
1& \text{if } \psi_i \leq \theta_i^*\\
0& \text{otherwise}
\end{cases}
\label{qi}
\end{equation}
\begin{equation}
p_i (\vec{\psi})=P_i(\psi_i)\coloneqq \psi_i Q_i (\psi_i )+\int_{\psi_i}^{\overline{{\theta}_i}} Q_i (s) \di s
\label{pi}
\end{equation}
\end{definition} 
In a simple direct mechanism, each data owner $s_i$ has a take-it-or-leave-it offer with $s_i$'s valuation bounded by a {\em threshold} $\theta_i^*$. If the reported $\psi_i$ is smaller than the threshold, the data owner will be selected and get a compensation $p_i (\vec{\psi})$, which is higher than $\psi_i$.  

\begin{lemma}\label{lem:IC&IR}
Assuming that $\theta_i^*$ is independent from the reported valuation $\psi_i$ for all $1\leq i\leq n$, a simple direct mechanism $\Psi$ is incentive compatible and individually rational.
\end{lemma}
\begin{proof}
For IR, suppose $\theta_i \leq \theta_i^*$. Then $Q_i(\theta_i)=1$. By \eqref{pi}, $P_i(\theta_i)$ equals
\begin{equation}\label{eqn:payment}
\theta_i Q_i(\theta_i)+\int_{\theta_i}^{\overline{\theta_i}}Q_i (s) \di s = \theta_i + \int_{\theta_i}^{\theta_i^*} 1 \di s=\theta_i^*,
\end{equation}
and  $U_i(\theta_i|\theta_i)=P_i(\theta_i)-\theta_i Q_i(\theta_i)=\theta_i ^* -\theta_i \geq 0$. If $\theta_i>\theta_i^*$, $Q_i(\psi_i)=0$ which implies $P_i(\theta_i)=0$ and $U_i(\theta_i|\theta_i)=0$. In either case, the expected utility of reporting the valuation truthfully is non-negative. 

The IC condition can be proved easily by considering the two possible cases of $\psi_i>\theta_i$ and $\psi_i<\theta_i$. The full proof is in Appendix B. 
\end{proof}

We would like to define a simple direct mechanism to solve \eqref{OP1}. 
It remains to find the appropriate threshold $\vec{\theta^*}$.  For the following lemma, by an {\em optimal threshold}, we mean a vector $\vec{\theta^*}=(\theta^*_1,\ldots,\theta^*_n)$ whose corresponding allocation and payment rules as defined in \eqref{qi} and \eqref{pi}, respectively, is an optimal solution for \eqref{OP1}.

\begin{lemma}\label{lem:theta}
The optimal solution to the following optimisation problem \eqref{OP2} is an optimal threshold.
\end{lemma}
\begin{equation}
\begin{aligned}
& \max_{\Theta}
& & \sum_{i=1}^n \varepsilon_i F_i (\theta_i^*) \\
& \text{such that }
&& \sum_{i=1}^n \theta_i^* F_i(\theta_i^*) = B \\
&
& & \underline{\theta} \leq \theta_i^* \leq \overline{\theta} \text{\qquad }\forall i \in \{1, \ldots, n\}
\end{aligned}
\label{OP2}
\end{equation}
\begin{proof} 
By substituting (\ref{Qi}) the objective function of \eqref{OP1} becomes $\sum_{i=1}^n \int_{\underline{\theta}}^{\overline{\theta}}
\varepsilon_i Q_i (\psi_i) f_i(\psi_i) \di\psi_i$, which, by \eqref{qi}, is $$\sum_{i=1}^n \int_{\underline{\theta}}^{\theta_i^*}
\varepsilon_i f_i(\psi_i) \di\psi_i=\sum_{i=1}^n \varepsilon_i F_i (\theta_i^*).$$ 

IC and IR are satisfied due to Lemma~\ref{lem:IC&IR} and the fact that $\theta^*_i$ is chosen by solving \eqref{OP2}, which is independent from $\psi_i$. 
BF is equivalent to $\sum_{i=1}^n\theta_i^* F_i(\theta_i^*)\leq B$ which can be derived using \eqref{Pi} and \eqref{pi}. Moreover, it is easy to see that (\ref{eqn:BF}) is binding, i.e., $\sum_{i=1}^n\theta_i^* F_i(\theta_i^*)=B$. Otherwise, we can always increase the value of $\theta_i^*$ and select more data owners. See full proof at Appendix C. 
\end{proof}

To solve problem \eqref{OP2}, take the Lagrange function 
\begin{multline}
L(\vec{\theta^*},\lambda,\mu_1,\ldots, \mu_n, \gamma_1, \ldots, \gamma_n)\coloneqq \sum_{i=1}^n \varepsilon_i F_i (\theta_i^* ) \\
-\lambda \sum_{i=1}^n (\theta_i^* F_i (\theta_i^*)-B)
- \mu_i (\theta_i^*-\overline{\theta})-\gamma_i (\underline{\theta}-\theta_i^* ) \\
\text{\qquad} \forall i \in \{1,\ldots, n\},
\end{multline}
 where $\lambda$, $\mu_i$ and $\gamma_i$ are Lagrange multipliers. Setting the first order derivative  to 0, we get:
 \begin{equation}\label{eqn:first order}
  \varepsilon_i f_i(\theta_i^*)-\lambda (F_i (\theta_i^*)+\theta_i^* f_i (\theta_i^*)) + \mu_i - \gamma_i =0 \text{\ } 
 \forall 1\leq i\leq n\\
\end{equation}
 The desirable threshold vector $\vec{\theta^*}$ is the solution to the system that contains \eqref{eqn:first order} and the following conditions:
\begin{equation}
\begin{aligned}
& \mu_i (\underline{\theta}-\theta_i^*)=0
&&  \forall i=1,\ldots,n \\
& \gamma_i (\theta_i^*-\overline{\theta})=0
&&  \forall i=1,\ldots,n \\
& \sum_{i=1}^n \theta_i^* F_i(\theta_i^*)-B= 0
&&   \\
& \lambda,\mu_1,\ldots, \mu_n, \gamma_1, \ldots, \gamma_n \geq 0
\end{aligned}
\label{KKT}
\end{equation} 



Our procurement mechanism $\Psi$ takes $\vec{\psi}$ and $\vec{\varepsilon}$ as inputs. It first solves the system above and obtains a threshold vector $\vec{\theta^*}$. 
$\Psi$ then selects data owners based on this vector: 
For $1\leq i\leq n$, select the data owner $s_i$ if $s_i$'s reported valuation $\psi_i$ is lower than $\theta_i^*$. In this case, make a payment of $\theta_i^*$ to $s_i$. Otherwise, $s_i$ is not chosen and the payment is 0. We propose an algorithm to implement the procurement mechanism as shown in Alg. \ref{alg:PM}.

\begin{algorithm}
\caption{Procurement mechanism $\Psi$}
\label{alg:PM}
\begin{algorithmic}
\STATE Solve the system \eqref{eqn:first order},\eqref{KKT} to obtain $\theta_i^*$ for $1\leq i\leq n$.
\FOR{$i \in \{1, \ldots, n\}$} 
    \IF{$\psi_i \leq \theta_i^*$}
        \STATE set $q_i\coloneqq 1$ and pay $p_i\coloneqq \theta_i^*$; 
    \ELSE
        \STATE set $q_i\coloneqq 0$ and pay $p_i\coloneqq 0$.
    \ENDIF
\ENDFOR
\end{algorithmic}
\end{algorithm}

The next theorem follows from Lemma~\ref{lem:theta}, Karush-Kuhn-Tucker theorem (see \cite{luenberger1997optimization}), and the convexity of Problem \eqref{OP2}. See full proof at Appendix D. 
\begin{theorem}
\label{Thr.IC}
The procurement mechanism $\Psi$ guarantees to find the optimal solution of Problem \eqref{OP1}. 
\end{theorem}

After data procurement, the data entries of selected data owners form a dataset $\vec{d}_\sampled$, and $P\E$ mechanism is applied on it. SMQ can meet the hard privacy constraints of all data owners. For each $s_i$, the achieved privacy is denoted as $\varepsilon_i'$. For those who are not selected, $\varepsilon_i’=0<\varepsilon_i$; for the selected ones, $\varepsilon_i’\leq \varepsilon_i$ is guaranteed by $P\E$ mechanism.  



\section{Experiment Setup}\label{sec:experiment}
Through the experiments, we aim to evaluate the performance of SMQ in terms of its query accuracy under different query types, budgets, and dependence relationships between $\vec{\theta}$ and $\vec{\varepsilon}$. 
We consider three query types, count, median and linear predictor. 
As performance metric, we compare the mean and $95 \% $ confidential interval (CI) of the returned query results against the true query answer $\varphi(\vec{d}_{\ground})$, and use root mean squared error (RMSE) to measure the error.



\noindent {\bf Datasets $\vec{d}_{\ground}$.} We use three real-world datasets, including Adults dataset\footnote{https://archive.ics.uci.edu/ml/datasets/Adult}, MovieLense 1M dataset \cite{harper2016movielens}, and Residential energy consumption survey (RECS) dataset \cite{energy}. Adults dataset consists of $32,561$ entries, each representing an adult living in the US \cite{dua2017uci}. It has $15$ attributes, including age, income, education, marital status, etc. 
The RECS dataset has $12,084$ records, each record with $940$ attributes, including identifier, region, division, etc. The MovieLense 1M dataset contains the information of $6,040$ audience, $3,952$ movies and $1,000,209$ ratings. 

\smallskip

\noindent {\bf Privacy parameters $\vec{\varepsilon}$ and data valuations $\vec{\theta}$.} The three datasets contain no information about privacy attitudes of the data owners, so we generate two sets of random numbers, representing data valuations $\theta_i$ and privacy requirements $\varepsilon_i$, respectively. $\varepsilon_i$ is a small non-negative number, hence, we restrict the privacy parameter to be bounded by $1$.  The $\vec{\varepsilon}$ and $\vec{\theta}$ are correlated uniformly distributed random variables in the range $(0,1)$. They are generated based on a correlation coefficient $\rho$. As a smaller $\varepsilon_i$ denotes a more stringent privacy requirement, negative values of $\rho$ expresses positive correlation. 
We consider three different dependence relationships between $\vec{\theta}$ and $\vec{\varepsilon}$: (1) {\em independence}, where $\rho\coloneqq 0$, (2) {\em partial positive correlation}, where $\rho\coloneqq -0.5$, and (3) {\em perfect positive correlation}, where $\rho\coloneqq -1$. In scenario (1), $\vec{\theta}$ and $\vec{\varepsilon}$ are irrelevant. In other words, having a high privacy requirement does not necessarily mean that this data owner attaches high valuation on her private data. In contrast, in scenarios (2) and (3), a data owner with high privacy requirement tends to have high data valuation.

\smallskip

\noindent {\bf Budget $B$.} The broker  has a budget $B \leq \overline{\theta}n$ for data procurement. We investigate the performance of the SMQ  under different budgets, $B=\{0.1 \overline{\theta}n, 0.2\overline{\theta}n, \ldots, 0.9\overline{\theta}n \}$.

\smallskip

\noindent {\bf Query types $\varphi$. }\label{sec:query} 
For count query, we use the income attribute of the Adults dataset,  gender attribute of the MovieLense 1M dataset and the total site electricity usage of the RECS dataset. The count queries ask: {\em How many adults have income higher than $50$ k? How many female audiences? And how many households consume more than 10 thousands kwh?}

For median query, 
we use the age attribute of the Adults dataset and the MovieLense 1M dataset and the total site electricity usage (integer-valued) of the RECS dataset. 

For the linear predictor, 
we use the data in the last row to represent $s_{n+1}$ and the data in the other rows to represent existing data owners. We choose the data of one attribute as $\vec{d}_{\ground}=(d_1, \ldots, d_n)$ and the data of the other attributes as the profile, denoted as $Y=(y_1, \ldots, y_n, y_{n+1})$, where $y_{n+1}$ is the profile of $s_{n+1}$, and used to calculate the similarities. Here, we use a common measure, {\em cosine similarity}, to quantify the similarity between $s_i$ and $s_{n+1}$, i.e., for each $i \in \{1,\ldots, n\}$, \(
w_i\coloneqq\cos\_sim(s_i,s_{n+1})=\frac{y_i \cdot y_{n+1}}{\| y_i\| \| y_{n+1}\|}
\).
We use the Adult dataset, the MovieLense 1M dataset and the RECS dataset, to predict whether a new individual's income is higher than $50$ k, whether a new individual likes the movie and whether a new household consume more than 10 thousands kwh, respectively.

In SMQ, the set $\Range_{\varphi}(\vec{d}_{\mathrm{c}})$ is constructed differently 
for different query types. For instance, $\Range_{\varphi}(\vec{d}_{\mathrm{c}})$ can be enumerated as $\frac{\sum_{i:s_i \in \S} q_i}{\sum_{i:s_i \in \S_c} q_i}t$ in a count query; as $t$ in a median query; and as $\frac{\sum_{i: s_i \in \S} w_i}{\sum_{i: s_i \in \S_c} w_i}t$ in a linear predictor, 
where $t \in \Range(\varphi)$ and $\S_c$ is the set of the selected data owners. 

\noindent {\bf Baselines.} For count and median queries, we compare SMQ  with FQ \cite{ghosh2011selling}. For linear predictor, we compare with FIP \cite{dandekar2012privacy}. 

\noindent {\bf FQ.} FQ achieves $\frac{1}{n-k}$-DP and pays equally for all $k$ selected data owners. 
FQ first computes privacy valuation $v_i\coloneqq \theta_i/\varepsilon_i$ for each $s_i$. It then selects the $k$ data owners with the least $v_i$ where $k$ is the largest integer satisfying $k v_k \leq B$. FQ then pays each selected data owner $\min \left\{B/k,v_{k+1}/(n-k)\right\}$ as compensation. The query answer is $r\coloneqq \sum_{i=1}^k d_i +(n-k)/2+Lap(n-k)$, where $Lap(n-k)$ is Laplace noise with variance $(n-k)$.

We also adjust the original FQ formulation for median queries. 
The mechanism follows the same allocation rule and pricing rule as those for count queries. The query answer is $r\coloneqq \varphi(\vec{d}_{\mathrm{c}})+Lap(\Delta \varphi(n-k))$, where $\Delta \varphi$ is the sensitivity of $\varphi$ on $\vec{d}_{\mathrm{c}}$ \cite{dwork2006calibrating}.


\noindent {\bf FIP.} 
FIP firstly sorts $v_i$ in ascending order. If there exists an $s_{i^*}$ whose weight $w_{i^*}$ satisfies $w_{i^*} > \sum_{i\colon \S \backslash \{s_{i^*}\}}w_i$, FIP only selects $s_{i^*}$. Otherwise, FIP selects the first $k$ data owners subject to $B$, i.e., $\frac{B}{\sum_{i=1}^k w_i}  \geq \frac{v_k }{\sum_{i=k+1}^n w_i}$ 
and pays $p_i=w_i \min \left\{\frac{B}{\sum_{i=1}^k w_i}, \frac{v_{k+1} }{\sum_{i=k+1}^n w_i}\right\}$ for each selected data owner. Assuming the range of the dataset is known as $[\underline{d},\overline{d}]$, FIP returns $r\coloneqq \sum_{i=1}^k w_i d_i + \frac{1}{2} (\underline{d}+\overline{d})\sum_{i=k+1}^{n} w_i + Lap\left( (\overline{d}-\underline{d}) \sum_{i=k+1}^{n} w_i\right)$.

As neither FQ nor FIP collects data owners' privacy requirements, to meet the privacy requirements in median and count queries, we need to ensure that FQ does not select data owners who have \(\varepsilon_i \leq 1/(n-k)\). Also, FIP sets the privacy level $\varepsilon_i \coloneqq w_i / \sum _{i:\in \S\backslash \S_c} w_i$, which is used for both SMQ and FIP, in order to make a fair comparison between these two mechanisms. 


For each query type, we test two mechanisms under different budgets and different dependence relationships between $\vec{\theta}$ and $\vec{\varepsilon}$. Under each experiment setup, $500$ trials are carried out, and the average, the $95 \% $ CI and the RMSE for each mechanism are reported. The experiment setups are summarised in Table~\ref{tab:setups}.

\begin{table}[h] \footnotesize
 \caption{Experiment Setups}
  \label{tab:setups}
  \begin{center}
  \begin{tabular}{ll}
    \midrule
    Query $\varphi$ & Count, median, linear predictor \\
    Dataset $\vec{d}_{\ground}$ & Adults, RECS, MovieLense 1M\\
    Data valuations $\vec{\theta}$ & $\vec{\theta} \sim U(0,1)$ \\
    Privacy parameters $\vec{\varepsilon}$ & $\vec{\varepsilon} \sim U(0,1)$ \\
    Correlation $\rho$ & $\{0, -0.5, -1\}$\\
    Budget $B$ & $\{0.1 \overline{\theta}n, 0.2\overline{\theta}n, \ldots, \overline{\theta}n\}$ \\
    Mechanism & SMQ, FQ, FIP mechanisms \\
    \bottomrule
  \end{tabular}
  \end{center}
\end{table}

\section{Results}
\noindent {\bf Experiment 1: Count query.}
We apply SMQ  and FQ to count queries on the three datasets. The two mechanisms demonstrate considerably different results on their allocation and payments. While FQ tends to select data owners with the lowest privacy valuations (regardless of their privacy requirements), SMQ tends to choose those data owners who have larger $\varepsilon_i$. This is because $\theta_i^*$ is determined by $\varepsilon_i$ and are different across data owners. Furthermore, FQ compensates the selected data owners uniformly while SMQ sets a price of $p_i=\theta_i^*$ which varies among data owners.

We compare the mechanisms in terms of accuracy under different budgets and different dependence relationships between $\vec{\theta}$ and $\vec{\varepsilon}$; See Fig.~\ref{fig:count}. As $B$  increases,  CI narrows down and RMSE displays a descending trend for both SMQ and FQ. SMQ outperforms FQ  in terms of accuracy for all datasets by a large margin. The CI for SMQ is significantly narrower and the RMSE for the SMQ is significantly lower than those for FQ across all cases.

The results also show that SMQ's performance improves as a stronger dependence relationship exists, whilst FQ's performance worsens, as CI becomes narrower and the RMSE becomes smaller for SMQ  when $\rho$ decreases from $0$ to $-0.5$ and $-1$. Those for FQ show an opposite trend. When \(\rho\) becomes smaller, the negative correlation between \(\theta_i\) and \(\varepsilon_i\) is higher. In other words, when \(\theta_i\) is large, \(\varepsilon_i\) is small, which makes \(v_i\) large. As a result, under the same budget, the number $k$ is smaller and the variance $n-k$ of Laplace noise is larger, which worsens the performance of FQ.

\begin{figure*}[t]
\centering
\includegraphics[width=10 cm]{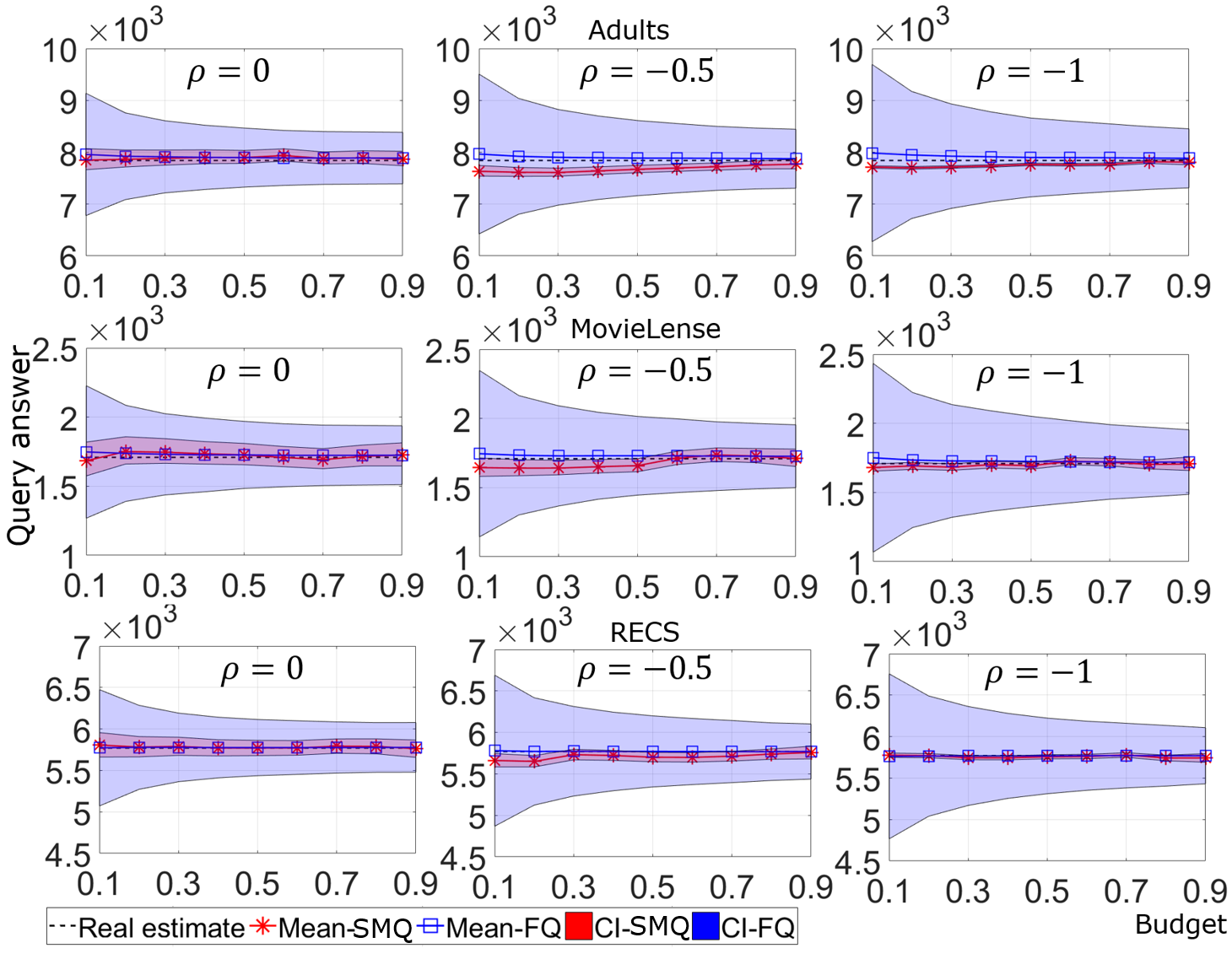}
\includegraphics[width=10 cm]{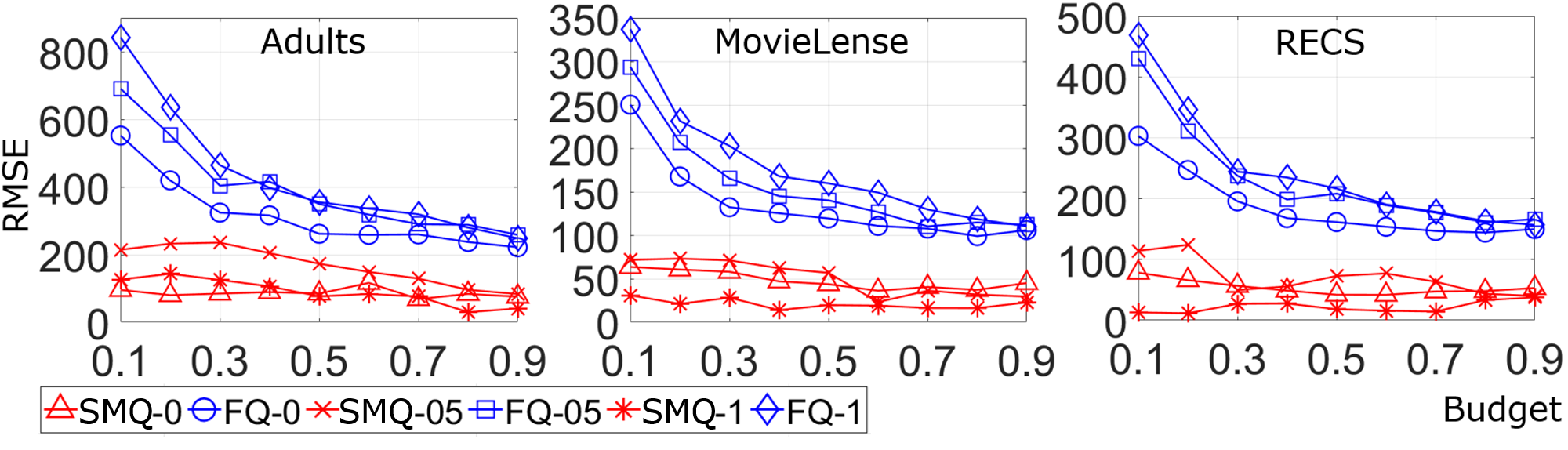}
\caption{\footnotesize CI and RMSE of SMQ and FQ for count query. The top three rows show the mean and $95 \%$ CI of the query answers of SMQ and FQ, where each row denotes a different dataset, and each column denotes a dependence relationship between $\vec{\theta}$ and $\vec{\varepsilon}$. 
The last row shows RMSE for different datasets. The horizontal axis indicates budget between $0.1\overline{\theta}n$ and $0.9\overline{\theta}n$. } 
\label{fig:count}
\end{figure*}

\noindent {\bf Experiment 2: Median query.}
We implement SMQ and FQ on median queries. 
The allocation and payment results are similar to those for count queries. In terms of the accuracy, SMQ significantly outperforms FQ for median queries. As shown in Figure \ref{fig:median}, the error for SMQ is negligible comparing to FQ. For the Adults and the MovieLense 1M datasets, the RMSE for the SMQ is zero. As for the RECS dataset, since the range is comparatively larger, the results are less accurate, but much better than the results for FQ. Also, under different value of $\rho$, SMQ  returns reliable results while FQ performs even worse when $\rho$ becomes smaller.

\begin{figure*}
\centering
\includegraphics[width=10 cm]{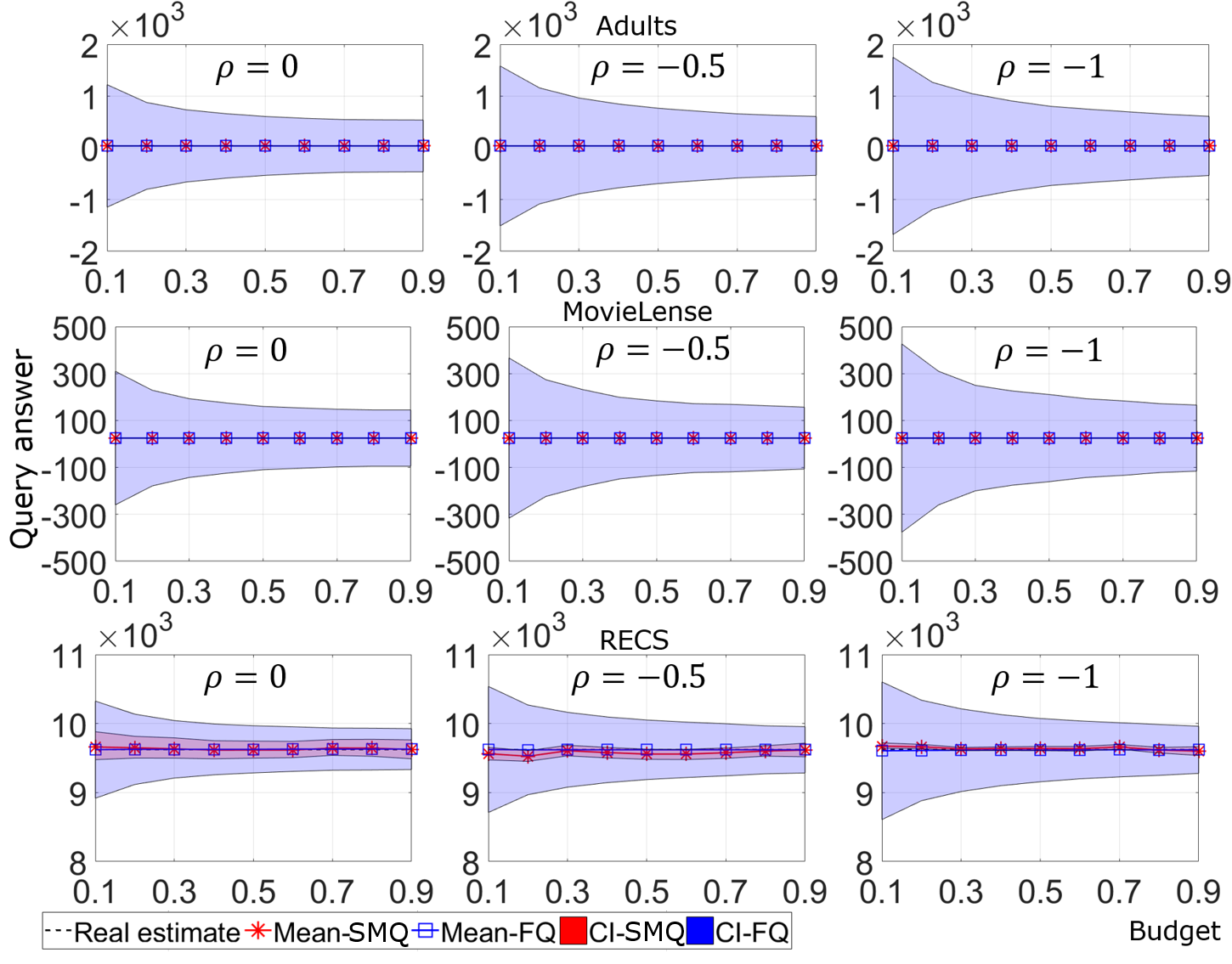}
\includegraphics[width=10 cm]{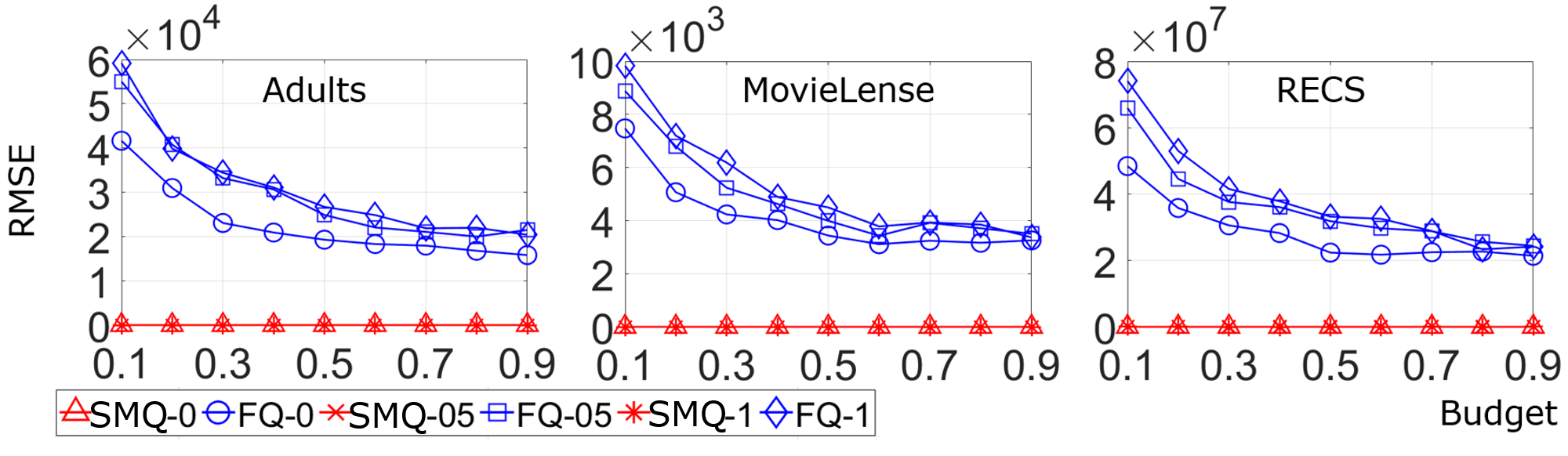}
\caption{\footnotesize CI and RMSE of SMQ and FQ for median query. All setting are the same as in Fig.~\ref{fig:count}.}
\label{fig:median}
\end{figure*}

\noindent {\bf Experiment 3: Linear predictor.}
We implement SMQ  and FIP for linear predictors. The allocation results of FIP is similar to FQ, where the data owners with low privacy valuations are chosen. The results show that for the Adult dataset, the number of data owners chosen by SMQ  is larger than that by FIP, while this number is similar for the other two datasets. 
Consistent with the previous experiments, SMQ  outperforms FIP in most cases with respect to accuracy, with the exception of when $B>0.5\overline{\theta} n$. As shown in Figure \ref{fig:lp}, for the Adult dataset, the CI is narrower and the RMSE is smaller for SMQ . For the MovieLense 1M dataset and RECS dataset, when the budget is low, the performance of SMQ  is better than FIP. When the the budget is high, SMQ outputs slightly less accurate results.

\begin{figure*} \centering
\includegraphics[width=10 cm]{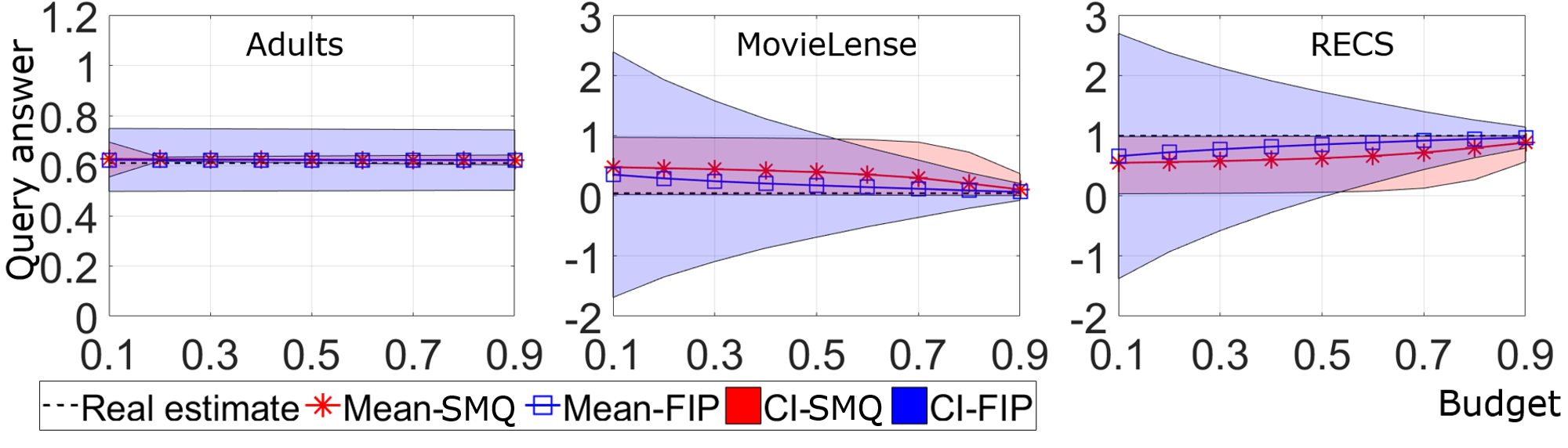}
\includegraphics[width=10 cm]{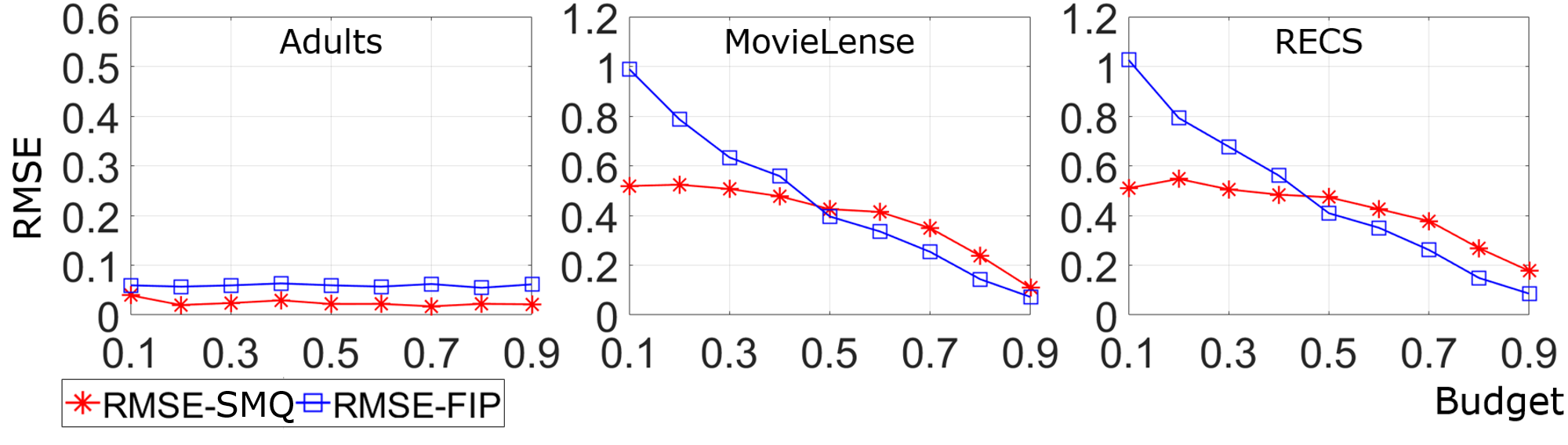}
\caption{\footnotesize CI and RMSE of SMQ and FIP for linear predictor. The top row shows the mean and $95 \%$ CI of the query answers of SMQ and FIP. The second row shows RMSE}
\label{fig:lp}
\end{figure*}

\section{Conclusion}\label{sec:conclusion} 
We consider private data query problem where the data are held by single-minded data owners. 
We propose the data query mechanism SMQ that satisfies IC, IR, BF and $\vec{\varepsilon}$-PDP for every $i \in \{1, \ldots, n\}$.  The empirical results show that SMQ effectively improves the query accuracy with the same budget than existing mechanisms.
%
An assumption in the work assumes that all data owners truthfully announce their privacy protection requirements. As future work, we will explore the data broker problem where hidden information comes from data valuation and privacy protection requirements. We also could study the problem with more sophisticated query types, e.g., queries whose output are beyond real values.   

\bibliographystyle{apacite}
\bibliography{short}

\newpage

\appendix

{\bf Appendix A.} 

\noindent {\bf Lemma \ref{accuracy}. }
For any integer $1\leq \alpha\leq n/4$ and $\delta\in (0,1)$, if the query mechanism $A$ is $\left(\alpha,\delta\right)$-PAC, then $\sum_{i=1}^n\varepsilon_i q_i\geq \frac{n}{4 \alpha}\cdot (\ln \delta-\ln(1-\delta))$. 

\medskip

\begin{proof} We prove the case when $\varphi$ is the count query. Recall that this case assumes that each data entry $d_i$ is a 0/1-value. 
We assume for a contradiction that $\sum_{i=1}^n\varepsilon_i q_i<\frac{n(\ln \delta-\ln(1-\delta))}{4\alpha}$ and the query mechanism is $(\alpha,\delta)$-PAC. Let $R=\{r \in \R\mid |r-\varphi(\vec{d}_{\ground})|<\alpha \}$. By the definition of $(\alpha,\delta)$-PAC,  
$\Pr\left(\Phi\left(\vec{d}_\ground\right)\in R \right) \geq \delta$.

Assume, w.l.o.g., that $\varepsilon_i q_i$ are sorted in ascending order, i.e., $\varepsilon_1 q_1\leq \varepsilon_2 q_2\leq \ldots \leq \varepsilon_n q_n$. Consider the first $ 4\alpha$ data owners (Note that $4\alpha\leq n$). 
Clearly, 
$$\sum_{i=1}^{4\alpha} \varepsilon_i q_i < \frac{n(\ln \delta-\ln(1-\delta))}{4\alpha}\frac{4\alpha}{n}=\ln \delta -\ln (1-\delta).$$
Let $\vec{d}^0 \coloneqq (d_{i})_{i\in I_0}$ and $\vec{d}^1 \coloneqq (d_i)_{i\in I_1}$ where $I_j= \{1\leq i\leq 4\alpha\mid d_i=j\}$ for $j\in \{0,1\}$. 
Without loss of generality, assume that $|\vec{d}^0|>2\alpha$. Let $I'\subseteq I_0$ that contains exactly $2\alpha$ elements, and define a dataset $\vec{d}'\coloneqq (b_1,\ldots,b_n)$ where $b_i=1$ if $i\in I'$, and $b_i=d_i$ otherwise.  
It follows that $\varphi(\vec{d}')=\varphi(\vec{d}_{\ground})+2\alpha$. 

It is straightforward to verify by definition of PDP that 
\begin{align*}
    \Pr\left(\Phi(\vec{d}')\in R\right) \geq & \exp\left(-\sum_{i\in I'}\varepsilon_i q_i\right) \Pr\left(\Phi(\vec{d}_{\ground})\in R\right)  \\
  >   & \exp\left(-(\ln\delta-\ln (1-\delta))\right)\times\delta   \\
  = & \frac{1-\delta}{\delta}\cdot \delta =1-\delta
\end{align*}
Since $\varphi(\vec{d}')=\varphi(\vec{d}_{\ground})+2\alpha $, by the triangle inequality, we have 
$\Pr\left(|\Phi(\vec{d}')-\varphi(\vec{d}')| > \alpha  \right) \geq \Pr\left(|\Phi(\vec{d}')-\varphi(\vec{d}_{\ground})| < \alpha  \right) > 1-\delta$, 
which contradicts the $(\alpha,\delta)$-PAC assumption. 

The proof is similar for the case when $\varphi$ is the general linear predictor where the data entries are real values. The only difference is that we define the set $I'$ as $\{1,\ldots,2\alpha\}$ and   
the dataset $\vec{d}'$ by $b_i=d_i+\frac{1}{w_i}$ for all $i\in I'$ and  $b_i=d_i$ otherwise. 



For the case when $\varphi$ is a median query. Assume $d_1,d_2,\ldots,d_n$ are distinct positive integers. 
We only deal with the  case when $n$ is odd (the case when $n$ is even can be proven in a similar way). 
 Let $m$ denote the median among $d_1,\ldots, d_n$. 
 Let $I_0\coloneqq \{i\mid d_i<m\}$ and $I_1\coloneqq \{i\mid d_i>m\}$. Suppose, w.l.o.g., that $\sum_{i\in I_0} \varepsilon_i q_i< \frac{n(\ln \delta-\ln(1-\delta))}{8\alpha}$. 
Let $k\coloneqq |\{i\mid m\leq d_i<m+2\alpha\}|$. Note that by mutual distinction of data values, $k\leq 2\alpha$. 
 For every $i\in I_0$, put $i$ into $H$ if the data owner $s_i$'s privacy requirement $\varepsilon_i$ is among the smallest $k$ among data owners in $I_0$. Clearly, $\sum_{i\in H} \varepsilon_iq_i \leq   \frac{n(\ln \delta-\ln(1-\delta))}{4\alpha}\frac{2\alpha}{n} < \ln \delta -\ln (1-\delta)$.
 Let $ d_{\max} \coloneqq \max\{d_1,\ldots,d_n\}$. Define a new dataset $\vec{d}' \coloneqq (b_1,\ldots,b_n) $ by $ b_i=d_i+d_{\max}$ if $i\in H$; and $b_i=d_i$ otherwise. It then follows that 
 the median of $\vec{d}'$ is at least $m+2\alpha$ and thus $\varphi(\vec{d}') \geq \varphi(\vec{d}_{\ground})+2\alpha$. 

By PDP of $\Phi$, we have $\Pr(|\Phi(\vec{d}')-\varphi(\vec{d}_{\ground})|<\alpha)> 1-\delta$. 
By the triangle inequality, we have 
$\Pr\left(|\Phi(\vec{d}')-\varphi(\vec{d}')| > \alpha  \right) \geq \Pr\left(|\Phi(\vec{d}')-\varphi(\vec{d}_{\ground})| < \alpha  \right) > 1-\delta$, which contradicts the accuracy assumption. 
\end{proof}

\noindent {\bf Appendix B.} 

\noindent {\bf Lemma \ref{lem:IC&IR}.}
Assuming that $\theta_i^*$ is independent from the reported valuation $\psi_i$ for all $1\leq i\leq n$, a simple direct mechanism $\Psi$ is incentive compatible and individually rational.

\bigskip

\begin{proof}
For IR, suppose $\theta_i \leq \theta_i^*$. Then $Q_i(\theta_i)=1$. By \eqref{pi}, $P_i(\theta_i)$ equals
\begin{equation}
\theta_i Q_i(\theta_i)+\int_{\theta_i}^{\overline{\theta_i}}Q_i (s) \di s = \theta_i + \int_{\theta_i}^{\theta_i^*} 1 \di s=\theta_i^*,
\end{equation}
and  $U_i(\theta_i|\theta_i)=P_i(\theta_i)-\theta_i Q_i(\theta_i)=\theta_i ^* -\theta_i \geq 0$. If $\theta_i>\theta_i^*$, $Q_i(\psi_i)=0$ which implies $P_i(\theta_i)=0$ and $U_i(\theta_i|\theta_i)=0$. In either case, the expected utility of reporting the valuation truthfully is non-negative. 

For IC, note that $\theta_i^*$ for all $i \in \{1,\ldots, n\}$ 
is independent from the reported valuation. When data owners report their valuations untruthfully, there are two cases:

\noindent Case (1) Suppose $s_i$ reports a valuation $\psi_i>\theta_i$.

a. if $\theta_i < \psi_i \leq \theta_i^*$,  $U_i(\psi_i|\theta_i)=U_i(\theta_i|\theta_i)=\theta_i^*-\theta_i$.

b. if $\theta_i \leq \theta_i^* < \psi_i$, $U_i(\theta_i|\theta_i)=\theta_i^*-\theta_i \geq 0=U_i(\psi_i|\theta_i)$.

c. if $\theta_i^*<\theta_i < \psi_i$, $U_i(\psi_i|\theta_i)=U_i(\theta_i|\theta_i)=0$.

\noindent Case (2) Suppose $s_i$ reports a valuation $\psi_i < \theta_i$.

a. if $\psi_i<\theta_i \leq \theta_i^*$, $U_i(\psi_i|\theta_i)=U_i(\theta_i|\theta_i)=\theta_i^*-\theta_i$.

b. if $\psi_i \leq \theta_i^* < \theta_i$, $U_i(\psi_i|\theta_i)=\theta_i^*-\theta_i<0=U_i(\theta_i|\theta_i)$.

c. if $\theta_i^*<\psi_i<\theta_i$, $U_i(\psi_i|\theta_i)=U_i(\theta_i|\theta_i)=0$.

The above argument shows that each data owner can maximise her expected utility by truthfully reporting the valuation.  
\end{proof}

\newpage

\noindent {\bf Appendix C.} 

\noindent {\bf Lemma \ref{lem:theta}}
The optimal solution to the optimisation problem \eqref{OP2} is an optimal threshold.

\begin{proof}

Firstly, since the threshold $\theta^*_i$ is determined by solving \eqref{OP2}, it is independent from $\psi_i$. By Lemma~\ref{lem:IC&IR},  IC and IR constraints are satisfied by allocation rule \eqref{qi} and payment rule \eqref{pi}.

For the objective function, by substituting (\ref{Qi}) the objective function becomes $\sum_{i=1}^n \int_{\underline{\theta}}^{\overline{\theta}}
\varepsilon_i Q_i (\psi_i) f_i(\psi_i) \di\psi_i$, which, by \eqref{qi}, is $$\sum_{i=1}^n \int_{\underline{\theta}}^{\theta_i^*}
\varepsilon_i f_i(\psi_i) \di\psi_i=\sum_{i=1}^n \varepsilon_i F_i (\theta_i^*).$$ For BF, by \eqref{Pi} the left hand side of the constraint \eqref{eqn:BF} is
\begin{align*}
&\sum_{i=1}^n \int_{\underline{\theta}}^{\overline{\theta}} P_i (\psi_i) f_i(\psi_i) \di \psi_i\\
=&\sum_{i=1}^n \int_{\underline{\theta}}^{\overline{\theta}} \left(\psi_i Q_i(\psi_i)+\int_{\psi_i}^{\overline{\theta}}Q_i(s) \di s\right)f_i(\psi_i) \di \psi_i& \text{by \eqref{pi}}\\
=&\sum_{i=1}^n \int_{\underline{\theta}}^{\theta_i^*} \theta_i^* f_i(\psi_i) \di \psi_i=\sum_{i=1}^n\theta_i^* F_i(\theta_i^*)&
\end{align*}
Thus \eqref{eqn:BF} is equivalent to $\sum_{i=1}^n\theta_i^* F_i(\theta_i^*)\leq B$. 
Moreover, it is easy to see that (\ref{eqn:BF}) is binding, i.e., $\sum_{i=1}^n\theta_i^* F_i(\theta_i^*)=B$. Otherwise, we can always increase the value of $\theta_i^*$ and select more data owners.
\end{proof}

\noindent {\bf Appendix D.}

\noindent {\bf Theorem \ref{Thr.IC}.}
The procurement mechanism $\Psi$ guarantees to find the optimal solution of Problem \eqref{OP1}. 

\smallskip

\begin{proof}
By Lemma~\ref{lem:theta}, we only need to show that the procurement mechanism $\Psi$ solves Problem \eqref{OP2}.
Define $B_i$ as $\theta_i^*F_i(\theta_i^*)$. The first constraint in \eqref{OP2} then becomes \(\sum_{i=1}^n B_i =B\), which is affine in terms of $B_i$. 

Also, since any \(B_i\) corresponds to a \(\theta_i^*\), we can view \(\theta_i^*\) as a function of \(B_i\) and thus write \(B_i=\theta_i^*(B_i)F_i(\theta_i^*(B_i))\). The derivative in terms of \(B_i\) is \[1=\theta_i^{*'}(B_i)F_i(\theta_i^*(B_i))+\theta_i(B_i)^*f_i(\theta_i^*(B_i))\theta_i^{* '}(B_i)\]
Reorganise the equation, we can get $$f_i(\theta_i^*)\theta_i^{*'}=\frac{1}{\frac{F_i(\theta_i^*)}{f_i(\theta_i^*)}+\theta_i^*}.$$ 
Because of the regularity assumption, the denominator is strictly increasing. Thus, \(f_i(\theta_i^*) \theta_i^{*'}\) is strictly decreasing. Furthermore, the derivative of the objective function in terms of \(B_i\) is $$\sum_{i=1}^n \varepsilon_i f_i(\theta_i^*(B_i))\theta_i^{*'}(B_i).$$ It is strictly decreasing as well. Therefore, the objective is to maximise a concave function. The above arguments asserts the convexity of  Problem \eqref{OP2}.

Since Problem \eqref{OP2} is  convex and the vector $\vec{\theta^*}$ satisfies conditions \eqref{eqn:first order} and \eqref{KKT}, Karush-Kuhn-Tucker theorem (see \cite{luenberger1997optimization}) implies that $\vec{\theta^*}$ is the optimal solution to \eqref{OP2}. 
\end{proof}
\end{document}